\documentclass[pdflatex,sn-mathphys-num]{sn-jnl}

\usepackage{graphicx}%
\usepackage{multirow}%
\usepackage{amsmath,amssymb,amsfonts}%
\usepackage{amsthm}%
\usepackage{mathrsfs}%
\usepackage[title]{appendix}%
\usepackage{xcolor}%
\usepackage{textcomp}%
\usepackage{manyfoot}%
\usepackage{booktabs}%
\usepackage{algorithm}%
\usepackage{algorithmicx}%
\usepackage{algpseudocode}%
\usepackage{listings}%

\theoremstyle{thmstyleone}%
\newtheorem{theorem}{Theorem}
\theoremstyle{thmstyletwo}%
\newtheorem{example}{Example}%
\newtheorem{remark}{Remark}%

\theoremstyle{thmstylethree}%
\newtheorem{definition}{Definition}%
\newtheorem{corollary}[theorem]{Corollary}%
\newtheorem{lemma}[theorem]{Lemma}%
\usepackage{xfrac}
\usepackage{relsize}

\usepackage[textsize=small]{todonotes}
\usepackage{xfrac}

\newcommand{\Fset}{\mathbb{F}}
\newcommand{\F}{\mathbb{F}}

\renewcommand{\vec}[1]{\text{\mathversion{bold}$#1$}}

\DeclareMathOperator{\rank}{rank}
\DeclareMathOperator{\diag}{diag}
\DeclareMathOperator{\circulant}{circ}
\DeclareMathOperator{\QR}{QR}
\DeclareMathOperator{\QNR}{QNR}
\DeclareMathOperator{\im}{Im}

\usepackage[normalem]{ulem}
\begin{document}

\title[Construction of Generalized Weighing-Hadamard Matrices over finite fields]{Construction of Generalized Weighing-Hadamard Matrices over Finite Fields}


\author[1]{\fnm{Gustavo} \sur{T. Bastos}}\email{gtbastos@ufsj.edu.br}
\equalcont{These authors contributed equally to this work.}

\author[2]{\fnm{Miguel} \sur{Beltrá}}\email{miguel.beltra@ua.es}
\equalcont{These authors contributed equally to this work.}

\author*[3]{\fnm{Sara} \sur{D. Cardell}}\email{sd.cardell@unesp.br}
\equalcont{These authors contributed equally to this work.}

\author[2]{\fnm{Verónica} \sur{Requena}}\email{vrequena@ua.es}
\equalcont{These authors contributed equally to this work.}

\affil*[1]{\orgdiv{Department of Mathematics and Statistics}, \orgname{Federal University of São João del-Rei}, \orgaddress{\street{170 Frei Orlando Square}, \city{São João del-Rei}, \postcode{36307-352}, \state{MG}, \country{Brazil}}}

\affil[2]{\orgdiv{Department of Mathematics}, \orgname{University of Alicante}, \orgaddress{\street{Carr. de San Vicente del Raspeig, s/n}, \city{San Vicente del Raspeig}, \postcode{03690}, \state{Alicante}, \country{Spain}}}

\affil[3]{\orgdiv{Department of Mathematics}, \orgname{Institute of Geosciences and Exact Sciences, Unesp}, \orgaddress{\street{Av. 24-A nº. 1515 }, \city{Rio Claro}, \postcode{13506-900}, \state{SP}, \country{Brazil}}}



 \abstract{

The existence, several properties, and constructions of Generalized Weighing-Hadamard (GWH) matrices over finite fields are addressed in this work. 
We study the subset of invertible GWH matrices and show that it forms a group under matrix multiplication. Besides that, we introduce a strong notion of equivalence between such matrices, defined via orthogonal transformations, and further prove that the corresponding quotient group by the subgroup of orthogonal matrices is abelian. Finally, we discuss some applications of these matrices in coding theory.}





\maketitle

\section{Introduction}

A Hadamard matrix $H$ of order $n$ is a square matrix with entries in $\{+1, -1\}$ such that
$
H H^T = n I_n,
$
meaning that the rows (and columns) of $H$ are mutually orthogonal. This implies that each row has Euclidean norm $\sqrt{n}$, and the inner product between any two distinct rows is zero. There are several well-known constructions of Hadamard matrices, which can be found in~\cite{Hedayat1978,Seberry1992,Huffman2021,Horadam2007}. These include Sylvester Hadamard matrices, Paley Type I and Type II Hadamard matrices, and Williamson Hadamard matrices.

Hadamard matrices play a fundamental role in various areas of applied mathematics, including optimal signal representations, information theory, signal processing, statistical analysis, and quantum computing \cite{bracewell2000fourier,Hadamard1893,Horadam2007,Macwilliams1977,nielsen2010quantum}. Specifically, in quantum computing, Hadamard matrices are central in the construction of quantum gates and error-correcting codes.

The existence of Hadamard matrices is subject to specific order constraints: it is well known that if a Hadamard matrix of order \(n>2\) exists, then \(n \equiv 0 \pmod{4}\). The \emph{Hadamard conjecture} asserts that a Hadamard matrix exists for every order divisible by \(4\); however, this remains open in general~\cite{Horadam2007}.

Beyond their algebraic properties, Hadamard matrices are also of significant interest in geometry and combinatorics. Indeed, they are closely related to error-correcting codes~\cite{harada,lin2004error} and symmetric block designs~\cite{Colbourn2006bk}. In particular, as an application to coding theory, we consider the notion of unbiased Hadamard matrices. Two Hadamard matrices \(H\) and \(K\) of order \(n\) are said to be unbiased if \( \frac{1}{\sqrt{n}}HK^{T}\) is itself a Hadamard matrix, or equivalently, if all entries of \(HK^{T}\) have absolute value \(\sqrt{n}\). This concept extends naturally to sets of mutually unbiased Hadamard matrices, and they can be used to construct ternary self-dual and self-orthogonal codes. Conversely, these codes can be used to generate new examples of unbiased Hadamard matrices (see  \cite{harada}).

Hadamard matrices have been generalized in several directions. One prominent example is the class of weighing matrices  \cite{Koukouvinos1997,Crnkovic2019}, 
which are square matrices of order $n$ with entries from $\{0,+1,-1\}$ satisfying $WW^T=kI_n$. Other significant generalizations include constructions over rings \cite{Sison2021} or finite fields \cite{Kojima2019,Kojima2025,Hurley2024}.

  In  \cite{Kojima2019,Kojima2025},  the authors  propose an extension of the concept of Hadamard matrices into prime  fields $\mathbb{F}_p$, where $p$ is an odd prime. They work with   square matrices over  $\mathbb{F}_p\setminus\{0\}$, where any pairs of rows are mutually orthogonal. Any additions and multiplications are executed modulo $p$. They referred these matrices as Hadamard-type (or H-type) matrices over $\mathbb{F}_p$.

In this work, we study a generalization of the weighing matrices and the Hadamard-type matrices   presented in~\cite{Kojima2025}, called Generalized Weighing-Hadamard (GWH) matrices over $\F_q$, which are square matrices of order $n$ satisfying $HH^T=wI_{n}$ for some $w \in \F_q$. We further investigate several properties and constructions of this class of matrices, including some adaptations of the constructions given in~\cite{Kojima2025}. 
Moreover, we provide recursive constructions that enable the generation of new GWH matrices from existing ones.
Finally, we establish the conditions for the existence of such matrices over arbitrary finite fields  and study its algebraic structure.



The paper is organized as follows. Section~\ref{sec:prel} reviews the necessary preliminaries. In Section~\ref{sec:GWH}, we study the notion of Generalized Weighing-Hadamard matrices and prove their existence over arbitrary finite fields. In Section~\ref{sec:con}, we analyze their properties and present several constructions. Section~\ref{sec:equiv} introduces an equivalence relation on the set of GWH matrices and establishes some of its algebraic properties. In Section~\ref{sec:ap}, we present several applications, to coding theory, illustrating the usefulness of GWH matrices. Finally, the paper concludes with some final remarks and directions for future research.

 \section{Preliminares}\label{sec:prel}

Consider $\Fset_q$ a finite field with $q=p^r$ elements, where $p$ is a prime number and $r$ a positive integer, and  $\mathbb{F}^{*}_q$
the multiplicative group of $\Fset_q$.
From now on, let $\mathsf{Mat}_{k \times n}(\Fset_q)$ and $\mathsf{Mat}_n(\Fset_q)$ be the $\mathbb{F}_q$-algebras of $k \times n$ and $n \times n $ matrices over $\F_q$, respectively. The matrices $I_n$ and $ O_{k \times n}$ denote the square identity and all-zero matrices, respectively. For simplicity, we just write $O$ to denote the zero matrix of the appropriate size.

\begin{definition}
An element $a \in \mathbb{F}_q$ is called a \emph{quadratic residue} if there exists $x \in \mathbb{F}_q$ such that
$$
x^2 = a.
$$
If there is no such $x$, then $a$ is called a \emph{quadratic non-residue}.
The set of quadratic residues of $\mathbb{F}_q$ is denoted by $\QR(q)$ and the set of quadratic non-residues by $\QNR(q)$.
\end{definition}

Observe that, as a consequence of the Frobenius automorphism  $\sigma: \mathbb{F}_{2^r} \longrightarrow \mathbb{F}_{2^r}$, given by $\sigma(a)=a^{2}$, all the elements in $\F_{2^r}$ are quadratic residues.

\begin{definition}
Let $\mathbb{F}_q$ be a finite field with $q$ elements, where $q$ is odd. 
The \emph{quadratic character} of an element $a \in \mathbb{F}_q^*$ is defined by $\chi: \mathbb{F}_q^* \longrightarrow \{-1,1\}$ such that $\chi(a) = a^{\frac{q-1}{2}}$. 
To extend, we can consider $\chi(a)=0$, for the case $a=0$.
\end{definition}
In \cite{Lidl1986bk}, authors introduce the  characterization of the quadratic residues through the quadratic character of $\Fset_q$, which we reflect in the following theorem. 
 \begin{theorem}
If $q$ is odd, an element $a \in \mathbb{F}_q^*$ is a quadratic residue if and only if
$
\chi(a)= 1;
$
and a quadratic non-residue if and only if
$
\chi(a) = -1.
$
Moreover, $|\QR(q)|=\frac{q-1}{2}=|\QNR(q)|$.
\end{theorem}

As a consequence of the previous theorem, we obtain some well-known properties of quadratic residues and non-residues over $\Fset_q$.

\begin{theorem} \label{th:quadratic_caracter}  
Let $ a,b \in \F_q ^*$. Then,
\begin{itemize}
    \item[(i)] $\chi(ab)=\chi(a)\chi(b)$, so the product of two quadratic residues or two quadratic non-residues is a quadratic residue, whereas the product of a residue and a non-residue gives a non-residue.
\item[(ii)] $\chi(a^{-1}) = \chi(a) ^{-1}$, so $a$ is a residue quadratic if and only if its inverse is a residue quadratic.
\item[(iii)] $\chi(a^n) = \chi(a) ^n$. 
\end{itemize}
\end{theorem}
As noticed in~\cite[p.~191]{finitefields}, for $q=p$ a odd prime and $c \in \F_p ^*$, the quadratic character $\chi(c)$ coincides with the Legendre Symbol $\left(\frac{c}{p}\right)$. 





 \section{Generalized Weighing-Hadamard Matrices over \texorpdfstring{$\Fset_q$}{Fq}}\label{sec:GWH}


In this section, we investigate the class of matrices known as Generalized Weighing-Hadamard matrices.
This family emerges as a natural synthesis of two fundamental structures in combinatorial design theory: the classical Hadamard matrices \cite{Hadamard1893,Hedayat1978}, characterized by their orthogonal rows, and the weighing matrices discussed in \cite{Koukouvinos1997,Crnkovic2019}, which relax the entry constraints while preserving a specific inner product property. 
First, we recall the classic definition of this family of matrices. 



 \begin{definition}
A \emph{Hadamard matrix} $H$ of order $n$ is a square matrix with entries in $\{+1, -1\}$ such that
$
H H^T = n I_n.
$
A \emph{weighing matrix} $W$ of order $n$ and weight $w$ is a square matrix with entries in $\{ -1, 0, 1 \}$ such that
\(
W W^T = w I_n.
\)
\end{definition}

Now, we are ready to introduce the definition of Generalized Weighing-Hadamard matrix, the main concept of this work.

\begin{definition} \label{Def:GWH}
Let $H \in \mathsf{Mat}_{n}(\Fset_q)$ a matrix of order $n$.
We say $H$ is a \emph{Generalized Weighing-Hadamard matrix}, or \emph{GWH matrix}, if there exists an element $w \in \Fset_q$ such that 
\[
HH^{T} = w I_n,
\]
where $w$ is called the \emph{weight} of the GWH matrix.
\end{definition}

\color{black}

 GWH matrices are related to what is known in the literature as \(\lambda\)-orthogonal matrices. If a GWH matrix
satisfies that $HH^T=H^TH=wI_n$, then it is known as a \(\lambda\)-orthogonal matrix. These are used in the construction of MDS, LCD, isodual, and self-dual codes~\cite{Mokhtari2026}. Observe that, if we consider GWH matrices with entries in $\{-1,0,1\}$, then we have weighing matrices~\cite{Crnkovic2019}. 
Moreover, in the case $w=n$ we obtain the Hadamard matrix if the entries are in $\{-1,1\}$; and, if we work over $\Fset_p$ we get the Hadamard-type matrices given in~\cite{Kojima2019,Kojima2025}. If $w=1$ then we have  an orthogonal matrix, which is a matrix $H \in \mathsf{Mat}_{n}(\Fset_q)$ such that $HH^{T}=I_n$.  


We now show that, under   some conditions, GWH matrices always exist. More precisely, we prove one of the main results of this paper: there exists a GWH matrix  $H \in \mathsf{Mat}_n(\Fset_q)$ of weight $w$ if and only if $w$ is a quadratic residue or $n$ is even. Before proving this result, we establish a couple of preliminary lemmas that will be needed in the proof.

\begin{lemma}\cite[Lemma 4.6]{Soler-Escriva2026}\label{lem:Xaro} Any element $w \in \mathbb{F}_q$ can be  expressed as
\( w = x^2 + y^2, \)
with $x,y \in \mathbb{F}_q$.
\end{lemma}

\begin{lemma}\label{thm:det_GWH}
    Let $H \in \mathsf{Mat}_n(\Fset_q)$ be a GWH matrix of weight $w$. Then $\det(H)^2 = w^n$.
\end{lemma}

\begin{proof} 
    $\det(H)^2 = \det(H) \det(H^T) = \det(HH^T) = \det(wI_n) = w^n$.
\end{proof}

\begin{theorem}[Existance of GWH matrices]\label{thm:existance} There exists a GWH matrix $H \in \mathsf{Mat}_n(\Fset_q)$ of weight $w$ if and only if $w$ is a quadratic residue or $n$ is even. 
\end{theorem}

\begin{proof}
   Suppose that $H \in \mathsf{Mat}_n(\mathbb{F}_q)$ is a GWH matrix of weight $w$.
 
By Lemma~\ref{thm:det_GWH}, we have $\det(H)^2 = w^n.$
 Since $\det(H) \in \mathbb{F}_q$, it follows that $w^n$ must be a square in $\mathbb{F}_q$, and, therefore, $w^n$ is a quadratic residue.
 Assume now that $w$ is a quadratic non-residue and that $n$ is odd.
Then, from Theorem~\ref{th:quadratic_caracter}
    \[ \chi(w^n) = \chi(w)^n = (-1)^n = -1. \]
It means that $w^n$ is a quadratic non-residue, which is a contradiction.
 Hence, $w$ is a quadratic residue  or $n$ is even.

Now, suppose that $w$ is a quadratic residue or $n$ is even.

If $w$ is a quadratic residue, then $H = \sqrt{w} I_n$ is clearly a GWH matrix of weight $w$.

If $n$ is even, suppose that $n=2k$ for some positive integer $k$. 
Lemma~\ref{lem:Xaro} implies that $w = a^2+b^2$, for some $a,b \in \Fset_q$. So, consider the matrix
\[
M =
\left[ 
    \begin{array}{rr}
        a & -b \\
        b &  a
    \end{array}
\right],
\]
then 
\[
MM^T =
\left[ 
    \begin{array}{cc}
        a^2+b^2 & 0 \\
        0 & a^2+b^2
    \end{array}
\right]
=
(a^2+b^2) I_2 = w I_2.
\]
Therefore, $M$ is a GWH matrix of weight $w$ of order $2$.
The result follows considering the block matrix
\[
H
= 
M \otimes I_k
=
\begin{bmatrix}
  M \\
  & M \\
  && \ddots \\
  &&& M
\end{bmatrix}
\in \mathsf{Mat}_{2k}(\Fset_q).
\]
\end{proof}

In particular, in fields of characteristic 2, there always exists a GWH matrix for any possible order and weight.

\begin{corollary}
Let $\mathbb{F}_q$ be a finite field of characteristic $2$ and let 
$w \in \mathbb{F}_q$.
Then, there always exists a matrix $H \in \mathsf{Mat}_{n}(\Fset_q)$ such that
$
HH^{T} = w I_n.
$
\end{corollary}

\begin{proof}
It is an immediate consequence of the Frobenius automorphism, since any element of a finite field $\Fset_q$ of characteristic $2$ is a quadratic residue.
\end{proof}

Notice that Theorem~4 and Corollary~1 of \cite{Kojima2025} are particular cases of the previous theorem; the case $w = n$ corresponds to the classical notion of a Hadamard-type matrix in the sense of authors in \cite{Kojima2025}.
We generalize those results in the following corollary.

\begin{corollary}[Existence of Hadamard-type matrices]\label{cor:HTM}
There exists a Hadamard-type matrix $H \in \mathsf{Mat}_n(\Fset_q)$ of weight $n$ if, and only if, $n$ is a quadratic residue or $n$ is even.   
\end{corollary}


\section{Constructions and Properties}
\label{sec:con}

In this section, we introduce some properties and constructions of the GWH matrices.
The first important result is that the product of GWH matrices is again a GWH matrix.

\begin{theorem}\label{thm:prod_of_GWH_matrices2}
    Let $H_1,\ldots, H_m \in \mathsf{Mat}_{n}(\Fset_q)$ be GWH matrices of weight $w_i$, for $i=1, \ldots, m$, respectively. Then,   the product $P_m=\prod_{i=1}^{m}H_i$  is a GWH matrix of weight $\prod_{i=1}^{m}w_i$.  
\end{theorem}

\begin{proof}
We prove it by induction on $m$. When $m=1$, the result is trivial. Assume
$$
P_{m-1} P_{m-1}^{T} = \left(\prod_{i=1}^{m-1} w_i\right) I_n.
$$
Let $P_m = P_{m-1} H_m$. Then
\[
P_m P_m^{T}
= P_{m-1} H_m (P_{m-1} H_m)^{T}= P_{m-1} H_m   H_m^{T} P_{m-1}^T.
\]
Using $H_m H_m^{T} = w_m I_n$, we get
\[
P_m P_m^{T}
= w_m \, P_{m-1} P_{m-1}^{T}
= \left(\prod_{i=1}^{m} w_i\right) I_n.
\]
\end{proof}

As an immediate consequence of  Theorem~\ref{thm:prod_of_GWH_matrices2}, we have the following corollaries.

\begin{corollary}
    \label{thm:prod_by_scalar}
    Let $H \in \mathsf{Mat}_{n}(\Fset_q)$ be a GWH matrix of weight $w$ and let $\alpha \in \Fset_q$ be an arbitrary element. Then $\alpha H$ is a GWH matrix of weight $\alpha^2 w$.
\end{corollary}
\begin{proof}
This result is an immediate consequence of Theorem \ref{thm:prod_of_GWH_matrices2} and the fact that the matrix $\alpha I_n \in \mathsf{Mat}_{n}(\Fset_q)$ is a GWH matrix of weight $\alpha^2$.
\end{proof}

\begin{corollary}\label{cor:prod_GWH_two_orthogonal}
    Consider $A,B \in \mathsf{Mat}_{n}(\Fset_q)$ orthogonal matrices and $H \in \mathsf{Mat}_{n}(\Fset_q)$ a GWH matrix of weight $w$. Then, $AHB$ is a GWH matrix of weight $w$.
\end{corollary}
\begin{proof}
    The proof is a consequence of Theorem \ref{thm:prod_of_GWH_matrices2} and the fact that orthogonal matrices are GWH of weight 1.
\end{proof}


Next corollary is an immediate consequence of Corollary~\ref{thm:prod_by_scalar}.

\begin{corollary}\label{cor:prod_by_square_root}
Let $w$ be a quadratic residue.
  \begin{itemize}
      \item[a)] If $A \in \mathsf{Mat}_n(\Fset_q)$ is an orthogonal matrix, then $H = \sqrt{w} A$ is GWH matrix of weight $w$.
      \item[b)] Let $H \in \mathsf{Mat}_n(\Fset_q)$ be a GWH matrix of weight $w \neq 0$. Then, $\frac{1}{\sqrt{w}} H$ is a orthogonal matrix, and hence, a GWH matrix of weight $1$.
  \end{itemize}
\end{corollary}

Since the class of GWH matrices is closed under matrix multiplication, it is natural to investigate whether it is also closed under addition. The following theorem shows that this requires additional conditions. 

 \begin{theorem}\label{thm:sum_of_GWH_matrices}
    Let $H_1,\ldots, H_m \in \mathsf{Mat}_{n}(\Fset_q)$ be GWH matrices of weight $w_i$, for $i=1, \ldots, m$, respectively, such that $H_iH_j^T=-H_jH_i^T$, $i\not = j$. Then,  the sum $S_m=\sum_{i=1}^{m}H_i$  is a GWH matrix of weight   $\sum_{i=1}^{m}w_i$.
\end{theorem}
\begin{proof} We prove it by induction on $m$. When $m=1$, it is trivial. Assume the result holds for $m-1$, i.e.,
\[
S_{m-1} S_{m-1}^{T} = \left(\sum_{i=1}^{m-1} w_i\right) I_n.
\]
Let $S_m = S_{m-1} + H_m$. Then
\[
S_m S_m^{T}=
( S_{m-1} + H_m)( S_{m-1} + H_m)^T
= S_{m-1} S_{m-1}^{T} 
+ S_{m-1} H_m^{T} + H_m S_{m-1}^{T}+ H_m H_m^{T}.
\]
By the inductive hypothesis and since $H_m H_m^{T} = w_m I_n$, we have
$$
S_m S_m^{T}= \left(\sum_{i=1}^{m} w_i\right) I_n
+ S_{m-1} H_m^{T} + H_m S_{m-1}^{T}.
$$
Now,
\[
S_{m-1}H_m^{T} +H_m  S_{m-1}^{T}
=\left(\sum_{i=1}^{m-1}H_i\right) H_m^{T} + H_m\left(\sum_{i=1}^{m-1}H_i\right)^T
= \sum_{i=1}^{m-1} \big( H_i H_m^{T} + H_m H_i^{T} \big) = O.
\]
Therefore,
\[
S_m S_m^{T} = \left(\sum_{i=1}^{m} w_i\right) I.
\]
 
\end{proof}

Next examples show some sets of matrices satisfying the conditions of Theorem~\ref{thm:sum_of_GWH_matrices}.

\begin{example} Let $\Fset_q$ be an arbitrary finite field. Consider the set of symmetric \emph{Pauli matrices}~\cite{nielsen2010quantum}, which are related to quantum information theory,
\[
    X =  \left[ 
        \begin{array}{cc}
            0 & 1 \\
            1 & 0
        \end{array}
    \right]
    \hspace{1cm} \text{and} \hspace{1cm}
    Z = \left[
        \begin{array}{cc}
            1 &  0 \\
            0 & -1
        \end{array}
    \right].
\]
It is easy to check that $XX^T = ZZ^T = I_2$ and $XZ^T = -ZX^T$. Hence, $X$ and $Z$ are GWH matrices of weight 1. Moreover, $X + Z$ is a GWH matrix of weight 2.
\end{example}

\begin{example} Let $\Fset_q$ be an arbitrary finite field and let $i = \sqrt{-1}$. Consider the set of symmetric \emph{Eddington matrices} \cite{Eddington1932} over the extension field $\Fset_q(i)$,
\begin{align}
E_1 & = 
\begin{bmatrix} 
  0 & i &  0 & 0 \\
  i & 0 &  0 & 0 \\ 
  0 & 0 &  0 & i \\
  0 & 0 &  i & 0
\end{bmatrix}, 
&
E_2 & =
\begin{bmatrix}
   i &  0 &  0 &  0 \\
   0 & -i &  0 &  0 \\
   0 &  0 &  i &  0 \\
   0 &  0 &  0 & -i
\end{bmatrix}, 
&
E_3 & =
\begin{bmatrix}
   0 &  0 &  0 & i \\
   0 &  0 & -i & 0 \\
   0 & -i &  0 & 0 \\
   i &  0 &  0 & 0
\end{bmatrix}.
\end{align}
These matrices satisfy $E_rE_r^T = -I_3$ for $r \in \{ 1,2,3 \}$ and $E_rE_s^T = -E_sE_r^T$ for $r,s \in \{ 1,2,3 \}$, $r \neq s$. Moreover, $E_r+E_s$ is a GWH matrix of weight $w = -2$ and $E_1+E_2+E_3$ is a GWH matrix of weight $-3$.
\end{example}

One of the most relevant consequences of Corollary \ref{thm:prod_by_scalar} is the fact that GWH matrices of non-zero weight are non-singular. Another important result is the fact that the property of being GWH matrix is preserved after transposition for non-singular GWH matrices.

\begin{theorem}
\label{thm:inverse_of_GWH_matrix}
     Let $H\in \mathsf{Mat}_{n}(\Fset_q)$ be a GWH matrix of weight $w \neq 0$. Then, $H$ is non-singular. Moreover, $H^{-1}$ and $H^{T}$ are GWH matrices of weight $w^{-1}$ and $w$, respectively.
\end{theorem}

\begin{proof} 
Definition~\ref{Def:GWH} implies that $H^{-1} = w^{-1} H^T$, so $H^T = wH^{-1}$. Thus,
\[ H^{-1} \left( H^{-1} \right)^T = H^{-1} \left( H^T \right)^{-1} = H^{-1} \left( wH^{-1} \right)^{-1} = w^{-1} H^{-1} H = w^{-1} I_n
. \] 

\noindent
In addition, $H^T (H^T)^T = H^T H = wH^{-1} H =  wI_n$, and the result holds.   
\end{proof}


\begin{remark}
Notice that if $w=0$, the transpose $H^T$ is not necessarily a GWH matrix of weight 0.
For example, the matrix
$H=
\left[
\begin{array}{cc}
     1 &2  \\
      1&2 
\end{array}
\right] 
$
is a GWH matrix of weight 0 in $\mathbb{F}_5$, since 
$$
HH^T=
\left[
\begin{array}{cc}
  1   &  2\\
 1    & 2
\end{array}
\right]
\left[
\begin{array}{cc}
  1   &  1\\
 2   & 2
\end{array}
\right]=
\left[
\begin{array}{cc}
  0  &  0\\
 0   & 0
\end{array}
\right].
$$
However,
$$
H^TH=
\left[
\begin{array}{cc}
  1   &  1\\
 2    & 2
\end{array}
\right]
\left[
\begin{array}{cc}
  1   &  2\\
 1   & 2
\end{array}
\right]=
\left[
\begin{array}{cc}
  2  &  4\\
4   & 3
\end{array}
\right].
$$
\end{remark}

Now, as consequence of Theorems~\ref{thm:prod_of_GWH_matrices2} and \ref{thm:inverse_of_GWH_matrix}, we have the following result. 
\begin{corollary}
    Let $H \in \mathsf{Mat}_{n}(\Fset_q)$ be a GWH matrix of weight $w\not=0$.
    The matrix $H ^m$, for $m\in\mathbb{Z}, $ is a GWH matrix of weight $w^m$.
\end{corollary}

We now show that the class of GWH matrices is closed under the action of the automorphism group of $\mathbb{F}_{q^r}$. Recall that any automorphism on $\Fset_{q^r}$ is a power of the Frobenius automorphism on $\Fset_{q^r}$.

\begin{theorem} 
 Let $H = [h_{ij}] \in \mathsf{Mat}_{n}(\Fset_{q^r})$ be a GWH matrix of weight $w$ and $\sigma$ an automorphism on $\Fset_{q^r}$. Then $\sigma(H) = [\sigma(h_{ij})] \in \mathsf{Mat}_{n}(\Fset_{q^r})$ is a GWH matrix of weight $\sigma(w)$.   
\end{theorem}

\begin{proof}
Due to the fact that $\sigma$ is an automorphism, we have that 
\[
\sigma({H})\sigma({H})^{T}=\sigma({H})\sigma({H}^{T})=\sigma({HH^{T}})=\sigma(wI)=\sigma(w)\sigma(I)=\sigma(w)I.
\]
\end{proof}

The following result introduces a construction of larger GWH matrices by applying the Kronecker product to a given collection of them.

\begin{theorem}\label{thm:Kronecker_product}
     Let $H_i \in \mathsf{Mat}_{n_i}(\Fset_q)$, for $i=1,2,\ldots,m$, be GWH matrices with weights $w_i$, respectively. Then, the Kronecker product
$H_1 \otimes H_2 \otimes \cdots \otimes H_m$ is a GWH matrix of order $n_1n_2\cdots n_m$ and weight $w_1w_2\cdots w_m$.
\end{theorem}
\begin{proof}
    From the properties of the Kronecker product, we have that
    \begin{align*}
    (H_1\otimes H_2\otimes \cdots \otimes H_m)(H_1\otimes H_2\otimes \cdots \otimes H_m)^T & = (H_1\otimes H_2 \otimes \cdots \otimes H_m)(H_1^T\otimes H_2^T \otimes \cdots \otimes H_m^T)\\
    & = (H_1 H_1^T) \otimes (H_2 H_2^T) \otimes \cdots  \otimes (H_m H_m^T)\\
    &= \left[ (w_1 I_{n_1}) \otimes (w_2 I_{n_2})  \otimes \cdots  \otimes (w_m I_{n_m})\right] \\
    & = w_1 w_2 \cdots w_m I_{n_1n_2\cdots n_m}
    \end{align*}
\end{proof}



  As a consequence of the previous theorem, we have the following result.
\begin{corollary}
    If $H \in \mathsf{Mat}_n(\Fset_q)$ is a GWH matrix of weight $w$, then 
    the matrix
    $$
    \underbrace{H \otimes H \otimes \cdots \otimes H}_{\text{$m$ copies}}
    $$
    is a GWH matrix of order $n^m$ and weight $w^m$.
\end{corollary}
\color{black}


The Hadamard product of two matrices  $A=[a_{ij}]$, $B=[b_{ij}]$ of order $n$, denoted by $A \odot B$, is the element-wise product of the matrices. More precisely, it is defined by
\[
A \odot B = [a_{ij} b_{ij}].
\]
The Hadamard product is defined only for matrices of the same size and differs from the standard matrix product, since it does not involve summation over indices.
The property of being GWH matrix is not preserved after the Hadamard product, as we show in the next example.

\begin{example} Consider $H_1, H_2 \in \mathsf{Mat}_{2}(\Fset_5)$  GWH matrices of weight $2$ given by
\begin{align*}
    H_1 & = \left[
        \begin{array}{cc}
            1 & 1 \\
            1 & 4
        \end{array}
    \right],  
    \quad
    H_2  = \left[ 
        \begin{array}{cc}
            4 & 4 \\
            4 & 1
        \end{array}
    \right].
\end{align*}
Their Hadamard product is
\[
    H_1 \odot H_2
    =
    \left[
        \begin{array}{cc}
            1 & 1 \\
            1 & 4
        \end{array}
    \right]
    \odot 
    \left[
        \begin{array}{cc}
            4 & 4 \\
            4 & 1
        \end{array}
    \right]
    =
    \left[
        \begin{array}{cc}
            4 & 4 \\
            4 & 4
        \end{array}
    \right],
\]
which is not a GWH matrix since
\[
    (H_1 \odot H_2)(H_1 \odot H_2)^T
    =
    \left[
        \begin{array}{cc}
            4 & 4 \\
            4 & 4
        \end{array}
    \right]
    \left[
        \begin{array}{cc}
            4 & 4 \\
            4 & 4
        \end{array}
    \right]^T
    =
    \left[
        \begin{array}{cc}
            2 & 2 \\
            2 & 2
        \end{array}
    \right].
\]
\end{example}

So far, we have discussed methods to derive GWH matrices through operations such as the product of GWH matrices, the sum of GWH matrices under specific restrictions, and the Kronecker product, we now shift our focus to the explicit construction of specific families of these matrices. 

\begin{theorem}\label{thm:GWH_matrix_times_generalized_orthogonal_matrix}
    Let $H \in \mathsf{Mat}_{n}(\Fset_q)$ be a GWH matrix of weight $w$ and $V \in \mathsf{Mat}_{n}(\Fset_q)$ a matrix such that $VV^T = \diag(\lambda_1,\lambda_2,\ldots,\lambda_n)$, where all $\lambda_i$ are non-zero quadratic residues, for $i=1,2,\ldots, n$.
    Consider the matrix 
    $$D :=
    \diag \left( \sqrt{\lambda_1^{-1}}, \sqrt{\lambda_2^{-1}}, \ldots, \sqrt{\lambda_n^{-1}} \right).$$ 
    Then, $DVH$ and $HDV$ are GWH matrices of weight $w$. 
\end{theorem}

\begin{proof}
Notice that $VV^T=(D^{-1})^2$. Then, 
    \[ \left(DVH \right) \left(DVH \right)^T = DVHH^TV^T D ^T = w DVV^TD = w D(D^{-1})^2D = wI_n. \]
    Analogously, $\left(HDV \right) \left(HDV \right)^T=wI_n$.
\end{proof}

\color{black}

Observe that Corollary~\ref{cor:prod_GWH_two_orthogonal}  is also a consequence of Theorem~\ref{thm:GWH_matrix_times_generalized_orthogonal_matrix}, where $\lambda_i = 1$ for all $i=1,2,\ldots,n$.
The next example illustrates Theorem~\ref{thm:GWH_matrix_times_generalized_orthogonal_matrix}.
\begin{example}\label{ex:ortogonal}
Let $H= \left[\begin{array}{ccc}
     1&2&2  \\
     2&1&3  \\
     2&3&1  \\
\end{array}
\right]\in \mathsf{Mat}_{3} \left(\mathbb{F}_{5}\right)$ be a GWH matrix of weight $w=4$ and $V= \left[\begin{array}{ccc}
     1&1&2  \\
     2&4&2  \\
     4&2&2  \\
\end{array}
\right]\in\mathsf{Mat}_{3} \left(\mathbb{F}_{5}\right)$. Notice that $VV^T=\diag(1,4,4)$ whose its entries are all quadratic residues, so consider $D=\diag(1,2,2) \in \mathsf{Mat}_{3} \left(\mathbb{F}_{5}\right)$. Then, 
\begin{equation*}
DVH= \left[\begin{array}{ccc}
     2&4&2  \\
     3&3&1  \\
     4&2&2  \\
\end{array}
\right] \mbox{ and } HDV= \left[\begin{array}{ccc}
     0&0&3  \\
     0&2&0  \\
     2&0&0  \\
\end{array}
\right]  
\end{equation*}
are both GWH matrices of weight 4.
\end{example}

The next theorems give us some methods to construct new GWH matrices of higher order starting with smaller ones. 
\color{black}
\begin{theorem} \label{th:blockGWH}
Let $w\neq 0$ and $w_i \neq 0$ be arbitrary elements of $\mathbb{F}_q$ such that  $\chi(w)=\chi(w_i)$, for $i=1,2, \ldots,m$. Define $v_i:=w w^{-1}_i$. If $H_1, H_2, \ldots, H_m$ are
    GWH matrices of order $n_i$ and weights $w_i$, respectively, then
        $$
    H=\left[
\begin{array}{cccc}
    \sqrt{v_1}H_1 &   &  & \\
       & \sqrt{v_2}H_2  & & \\
     &   &   \ddots   &\\
    & & &\sqrt{v_m}H_m\\
\end{array}
    \right],
    $$
    is a GWH matrix of order $n_1+n_2+\cdots+n_m$ and weight $w$.
\end{theorem} 
\begin{proof}
    We have that 
    \begin{align*}
    HH^{T}& = 
    \left[
\begin{array}{cccc}
    \sqrt{v_1}H_1 &    &  & \\
       & \sqrt{v_2}H_2  & & \\
     &   &  \ddots    &\\
    & & &\sqrt{v_m}H_m\\
\end{array}
    \right]\left[
\begin{array}{cccc}
    \sqrt{v_1}H_1^T &    & & \\
      & \sqrt{v_2}H_2^T  & & \\
     &   & \ddots  &\\
    & & &\sqrt{v_n}H_m^T\\
\end{array}
    \right]\\
    &=
    \left[
\begin{array}{cccc}
    v_1H_1H_1^T &    &  &\\
      & v_2H_2H_2^T  & & \\
     &   & \ddots   &\\
    & & &v_mH_mH_m^T\\
\end{array}
    \right] = 
    \left[
\begin{array}{cccc}
    v_1w_1I_{n_1} &    && \\
       & v_2w_2I_{n_2}  && \\
     &   & \ddots      &\\
    & & &v_mw_mI_{n_m}\\
\end{array}
    \right]\\
    & =wI_{n_1+n_2+\cdots+n_m}
    \end{align*}
\end{proof}
\begin{theorem}
\label{th:comp}
Let \(H_1, H_2, H_3 \in \mathsf{Mat}_n(\F_q)\) be GWH matrices such that \(H_1\) and \(H_3\) have weight \(w \neq 0\), while \(H_2\) has weight \(w' \neq 0\).
Then
\[
    H=\left[
    \begin{array}{cccc}
        H_1 & H_2 \\
        -\frac{1}{w} H_3 H_2^T H_1  & H_3
    \end{array}
    \right],
\]
is a GWH matrix of order $2n$ and weight $w+w'$.
\end{theorem}
\begin{proof}
\begin{align*}
    HH^T
    & =
    \left[
    \begin{array}{cccc}
        H_1 & H_2 \\
        -\frac{1}{w} H_3 H_2^T H_1  & H_3
    \end{array}
    \right]
    \left[
    \begin{array}{cccc}
        H_1^T & -\frac{1}{w} H_1^T H_2 H_3^T \\
        H_2^T & H_3^T
    \end{array}
    \right]
    \\
    & =
    \left[
    \begin{array}{cccc}
        H_1H_1^T + H_2H_2^T & O \\
        O & \frac{1}{w^2} (H_3H_2^TH_1)(H_3H_2^TH_1)^T + H_3H_3^T
    \end{array}
    \right]
    \\
    & =
    \left[
    \begin{array}{cccc}
        (w+w') I_n & O \\
        O & (\frac{w^2w'}{w^2} + w) I_n
    \end{array}
    \right]
    =
    (w+w')
    \left[
    \begin{array}{cccc}
        I_n & O \\
        O & I_n
    \end{array}
    \right].
\end{align*}
\end{proof}

In the next   theorem, we aim to construct a triangular version of Theorem~\ref{th:comp}. However, to ensure that the blocks are non-zero, the matrices must have weight 0, as shown below.

\begin{theorem}
Let $H_1, H_2, H_3 \in \mathsf{Mat}_n(\mathbb{F}_q)$ be non-zero GWH matrices of weights 
$w_1, w_2,$ and $w_3$, respectively. 
The block matrix
\[
H=
\begin{bmatrix}
H_1 & H_2\\
O& H_3
\end{bmatrix}
\]
  is a GWH matrix if and only if 
$
H_2H_3^T=O
$
and $w_1=w_2=w_3=0$.
As a consequence, $H$ is also a  GWH matrix of weight 0.
 
\end{theorem}
\begin{proof}
We have that
\begin{equation}\label{eq:triangular_block_matrix}
HH^T=
\begin{bmatrix}
H_1 & H_2\\
O & H_3
\end{bmatrix}
\begin{bmatrix}
H_1^T & O\\
H_2^T & H_3^T
\end{bmatrix}
=
\begin{bmatrix}
H_1H_1^T + H_2H_2^T & H_2H_3^T\\
H_3H_2^T & H_3H_3^T
\end{bmatrix}.
\end{equation}
Since $H_1,H_2,H_3$ are GWH matrices, we have
\[
H_1H_1^T=w_1I_n,\qquad
H_2H_2^T=w_2I_n,\qquad
H_3H_3^T=w_3I_n.
\]
Hence, 
\[
HH^T=
\begin{bmatrix}
(w_1+w_2)I_n & H_2H_3^T\\
H_3H_2^T & w_3I_n
\end{bmatrix}.
\]
For the direct implication, suppose that $H$ is a GWH matrix. Then, it must be satisfied that
\[
w_3 = w_1+w_2, \quad H_3H_2^T = O, \quad H_2H_3^T = O.
\]
Now, observe that
$$
O = H_2^T (H_2H_3^T) = (H_2^TH_2) H_3^T = w_2H_3^T,
$$
so if $w_2\not = 0$ then $H_3^T=O$,
which is a contradiction. Necessarily, $w_2=0$.
Analogously,
$$
O = (H_2H_3^T)H_3 = H_2(H_3^TH_3) = w_3H_2,
$$
so if $w_3 \neq 0$ then, $H_2=O$, which is again a contradiction. Necessarily $w_3=0$, and therefore $w_1=w_3-w_2=0$.

For the converse implication, suppose that $w_1=w_2=w_3=0$ and $H_2H_3^T=O$. Thus, $H_1H_1^T = H_2H_2^T = H_3H_3^T = O$. Notice also that $H_2H_3^T=O$ implies $H_3H_2^T=O$. From these facts it is easy to check that $HH^T=O$ by looking at \eqref{eq:triangular_block_matrix}.
\end{proof}

The next example shows that there exist matrices satisfying the conditions of the previous theorem.
\begin{example}
    Consider the block matrix over $\mathbb{F}_5$ given by
    $$
H=
\left[
\begin{array}{c|c}
 H_1    & H_2  \\\hline
  O   & H_3 
\end{array}
\right]
=
\left[
\begin{array}{cc|cc}
3 & 4 & 2 & 1\\
3 & 4 & 2 & 1\\\hline
0 & 0 & 2 & 1 \\
0 & 0 & 2 & 1
\end{array}
\right],
$$
Observe that the matrices $H_i$, for $i=1,2,3$, are GWH of weight $0$.  
It is easy to check that $HH^T=O$.

\end{example}

We now present constructions of circulant and bordered circulant GWH matrices. These constructions provide a systematic way to generate large families of examples with prescribed orthogonality properties.

\begin{theorem}\label{GWHcircmatric}
A circulant matrix $C=\circulant(c_0,c_1,\ldots, c_{n-1}) \in \mathsf{Mat}_n(\mathbb{F}_q)$ is a GWH matrix of weight $w$ if and only if the following conditions hold: \begin{itemize}
    \item [a)] $\sum_{k=0}^{n-1} c_k^2 = w$, and
    \item [b)] $\sum_{k=0}^{n-1} c_k c_{k+d} = 0$ for all $d=1,\ldots,\left\lfloor\frac{n}{2}\right\rfloor$.
\end{itemize}
\end{theorem}

\begin{proof}
Since $C$ is circulant then $C^T$ is circulant, and hence $D=CC^T$ is also circulant. Thus, it suffices to describe only the first row of $D$ in order to determine all its entries.
Let $ D= \circulant(d_0 , d_1 ,\ldots, d_{n-1})$. Then
\[
d_\ell =  \sum_{j=0}^{n-1} c_j \, c_{j+\ell},
\]
 for each $\ell \in \{0, 1, \ldots, n-1\}$, where the indices are taken modulo $n$. 
Indeed, for any $0 \le j \le n-1$, each term $c_jc_{j+r}$ appearing in $d_r$ also appears in $d_{n-r}$. In fact, the $(j+r+1 \pmod n)$-th term of $d_{n-r}$ is
\[
c_{j+r}c_{j+r+n-r}=c_{j+r}c_j=c_jc_{j+r},
\]
which coincides with the $(j+1 \pmod n)$-th term of $d_r$. Consequently,
\[
d_1 = d_{n-1},\quad
d_2 = d_{n-2},\quad \ldots,\quad
d_{\left\lfloor \frac{n}{2}\right\rfloor}
=
d_{\left\lceil \frac{n}{2}\right\rceil}.
\]
Therefore, $C$ is a GWH matrix of weight $w$, that is,
$D=CC^T = wI_n,$
if and only if the following conditions hold for every $0 \le k \le n-1$,
\begin{itemize}
    \item [a)] $\sum_{k=0}^{n-1} c_k^2 = w$, and
    \item [b)] $\sum_{k=0}^{n-1} c_k c_{k+d} = 0$ for all $d=1,\ldots,\left\lfloor\frac{n}{2}\right\rfloor$.
\end{itemize}
\end{proof} 

Based on Theorem~\ref{GWHcircmatric}, families of circulant matrices provide a rich source of GWH matrices with explicit computational advantages, since all rows of a circulant matrix have the same norm
and it suffices to verify the orthogonality of the first row with the next $\left\lfloor \frac{n}{2}\right\rfloor$ rows.

\begin{example} 
Consider the following circulant matrix
\[
C= \circulant(0,1,2,2,4)=
\begin{bmatrix}
0 & 1 & 2 & 2 & 4\\
4 & 0 & 1 & 2 & 2\\
2 & 4 & 0 & 1 & 2\\
2 & 2 & 4 & 0 & 1\\
1 & 2 & 2 & 4 & 0
\end{bmatrix}\in \mathsf{Mat}_{5}(\F_7).
\]
Since
\[
w=d_0=\sum_{i=0}^{4} c_i^2 =4,
\qquad
d_1=\sum_{i=0}^{4} c_i c_{i+1}=0,
\qquad
d_2=\sum_{i=0}^{4} c_i c_{i+2}=0,
\]
it follows that
$
D=CC^T=4I_5.
$
Hence, $C$ is a GWH matrix of weight $4$.
\end{example}

%

    

\begin{definition}
A matrix $A$ is \textbf{bordered circulant} if it is of the form
\[
A=\begin{bmatrix}
\alpha & \beta & \beta & \cdots & \beta \\ 
\gamma &  &  &  &  \\ 
\vdots &  &  & C &  \\ 
\gamma &  &  &  & 
\end{bmatrix},
\]
where $C$ is a circulant matrix.    
\end{definition}



Bordered double circulant constructions play an important role in the construction of isodual codes \cite{Bachoc2000},
meaning that the code is equivalent to its dual under a suitable transformation. This property is of special interest in the study of self-orthogonal and LCD codes, as well as in applications to cryptography and quantum error-correcting codes, where duality constraints are essential.

\color{black}

\begin{theorem}
A bordered circulant matrix
$$
A=\begin{bmatrix}
\alpha & \beta & \beta & \cdots & \beta \\ 
\gamma &  &  &  &  \\ 
\vdots &  &  & C &  \\ 
\gamma &  &  &  & 
\end{bmatrix},
$$
where $C=\circulant(c_0,c_1, \ldots, c_{n-1})$, is GWH of weight $w$ if and only if the following properties are satisfied:
\begin{enumerate}
    \item[a)] $w=\alpha^{2}+n\beta^{2}$,    
    \item[b)] $\alpha\gamma+\beta \sum_{i-0}^{n-1}c_i=0,$ and

    \item[c)] $CC^{T}=wI_{n}-\gamma^{2}J$, where \(J\) is the \(n\times n\) all-one matrix.
\end{enumerate}
\end{theorem}

\begin{proof} Let us denote by $s=\sum_{i=0}^{n-1}c_i$ the sum of each row of $C$ and by $J$ the $n\times n$ all-one matrix. Then
\begin{align*}
AA^{T} & =
\begin{bmatrix}
\alpha^2+n\beta^2 &\alpha\gamma+\beta s &\alpha\gamma+\beta s &\cdots &\alpha\gamma+\beta s\\
\alpha\gamma+\beta s & &\\
\alpha\gamma+\beta s & &\multirow{4}{*}{$CC^T+\gamma^2J$} \\
\vdots & & \\
\alpha\gamma+\beta s && 
\end{bmatrix}.
\end{align*}
Thus, $AA^T = wI_{n+1}$ for some $w \in \mathbb{F}_q$, if and only if the following conditions hold:
\begin{align*}
    \alpha^{2}+n\beta^{2} & = w, \\
    \alpha\gamma + \beta s & = 0, \\
    CC^{T} + \gamma^{2}J & = wI_{n},
\end{align*}
or equivalently, if and only if, conditions a), b) and c) of the theorem hold.
\end{proof}




\color{black}
\begin{corollary}\label{thm:alpha}
Let $\alpha\in \F_q ^*$, $\gcd(n+1,q)=1$, $ C=\circulant(c_0,c_1,\ldots,c_{n-1}) $ and consider
\begin{equation*}
M = \left[\begin{array}{cccc}
    \alpha & \alpha & \cdots & \alpha  \\
    \alpha &&& \\
    \vdots && C & \\
    \alpha &&& 
\end{array}\right] \in \mathsf{Mat}_{n+1}(\F_q).
\end{equation*}
Then, $M$ is GWH matrix of weight $w\not = 0$ if and only if
\begin{enumerate}
    \item [a)] $(n+1)\alpha^2 =w$,
    \item [b)] $\sum_{k=0} ^{n-1} c_k =-\alpha$
    \item [c)] $CC^T=wI_n-\alpha^2J_n$, where \(J_n\) denotes the all-one matrix of order \(n\).
\end{enumerate}
\end{corollary}

\color{black}

\begin{remark} Assume that the matrix $M$ in Corollary~\ref{thm:alpha} is a GWH matrix.
    \begin{enumerate}
        \item If $\gcd(n+1,q)\not=1$ or $\alpha=0$, then the matrix $M$ is of weight $w=0$.
   
       \item Note that, in the particular case where $q$ is prime, $\alpha=1$, and $c_{1}=c_{2}=\cdots=c_{n-1}$, the matrix $M$ in Corollary~\ref{thm:alpha} is a Standard-form Cyclic H-type matrix as defined in~\cite[Def.~2]{Kojima2025}. In this case, $MM^T=(n+1)I_{n+1}$, that is, $M$ is a GWH matrix of weight $n+1$ and order $n+1$.
    \end{enumerate}
\end{remark}

\begin{example}
Let us construct a GWH matrix of weight $4$ over $\mathbb{F}_{11}$.
Consider the values $\alpha = 10$ and $C = \circulant(10,1,1)$. It holds:
\[ (n+1)\alpha^2 = 4 \cdot 10^2 = 4, \]
\[ \sum_{k=0}^{n-1} c_k = 10+1+1 = -10. \]
Moreover,
\[
    CC^T
    =
    \begin{bmatrix}
        10 & 1 & 1 \\
        1 & 10 & 1 \\
        1 & 1 & 10
    \end{bmatrix}
    \begin{bmatrix}
        10 & 1 & 1 \\
        1 & 10 & 1 \\
        1 & 1 & 10
    \end{bmatrix}
    =
    \begin{bmatrix}
         3 & 10 & 10 \\
        10 &  3 & 10 \\
        10 & 10 &  3
    \end{bmatrix}
\]
and 
\[
    wI_n-\alpha^2 J_n
    =
    \begin{bmatrix}
        4 & 0 & 0 \\
        0 & 4 & 0 \\
        0 & 0 & 4
    \end{bmatrix}
    -
    1 \cdot
    \begin{bmatrix}
        1 & 1 & 1 \\
        1 & 1 & 1 \\
        1 & 1 & 1
    \end{bmatrix}
    =
    \begin{bmatrix}
         3 & 10 & 10 \\
        10 &  3 & 10 \\
        10 & 10 &  3
    \end{bmatrix}.
\]
Since $CC^T = wI_n - \alpha^2 J_n$, the matrix
\[
H =
\begin{bmatrix}
    10 & 10 & 10 & 10 \\
    10 & 10 &  1 &  1 \\ 
    10 &  1 & 10 &  1 \\
    10 &  1 &  1 & 10
\end{bmatrix}
\]
is a GWH matrix of weight $w = 4$.
\end{example}

\section{GWH-equivalences} \label{sec:equiv}



In this section, we define two binary relations between GWH matrices based on Corollary~\ref{cor:prod_GWH_two_orthogonal}. We show that GWH matrices of the same non-zero weight are equivalent. We study the quotient defined by the equivalence relation and show that it has interesting algebraic properties.

\begin{definition}\label{def:eq}
Consider $H, K \in \mathsf{Mat}_{n}(\Fset_q)$ GWH matrices.
We say that $H$ is \emph{GWH-equivalent} to $K$, denoted by $H \sim K$,  if there exist orthogonal
matrices $P, Q \in \mathsf{Mat}_{n}(\Fset_q)$ such that $K = PHQ$. In addition, we say that $H$ is \emph{strongly GWH-equivalent} to $K$, denoted by $H \sim_{s} K$, if there exist an orthogonal
matrix $Q \in \mathsf{Mat}_{n}(\Fset_q)$ such that $K = HQ$.
\end{definition}

Observe that both the GWH-equivalence and strongly GWH-equivalence are equivalence relations.  Moreover, strongly GWH-equivalence implies GWH-equivalence as we can take $P=I_n$.
In particular, since $P$ and $Q$ are invertible, GWH-equivalent matrices are also equivalent in the classical sense of matrices. 
Recall that two matrices of the same size are equivalent if and only if they have the same rank \cite[Corollary 2.8]{Hefferon2020}, 
so it follows that the rank remains invariant under the GWH-equivalence.
\begin{example}
    Consider the matrices 
    $$
    H=
 \left[\begin{array}{ccc}
     3& 0& 0\\
     0 &9 & 4\\
     0&7&9
\end{array}
\right]\quad \text{and}\quad 
    K=
 \left[\begin{array}{ccc}
   0 &  4 &  9\\
   0 &  9  & 7\\
   3  & 0  & 0\\
\end{array}
\right].
$$ 
Both matrices are GWH of weight 9 over $\mathbb{F}_{11}$.
It is easy to check that $K=PHQ$, where $P$ and $Q$ are the orthogonal matrices
$$
P =
\left[
\begin{array}{ccc}
   0 &  1 &  0\\
   0 &  0 &  1\\
   1 &  0 &  0\\
   \end{array}
      \right],
   \quad 
   Q =
\left[
\begin{array}{ccc}
   1 &  0 &  0\\
   0 &  0 &  1\\
   0  & 1 &  0\\
   \end{array}
   \right].
$$
Therefore, $H\sim K$.
\end{example}

Corollary~\ref{cor:prod_GWH_two_orthogonal} shows that the GWH-equivalence and the strongly GWH-equivalence preserve the weight. Moreover, when restricting our attention to non-zero weight GWH matrices, it turns out that the weight is precisely the key feature of these equivalence relations, as it is shown in the next theorem. 
\begin{theorem}\label{th:equiv}
    Let $H, K \in \mathsf{Mat}_{n}(\Fset_q)$ GWH matrices of non-zero weights $w_1, w_2$, respectively. Then, $H$ and $K$ are strongly GWH-equivalent if and only if $w_1=w_2$.
\end{theorem} 
\begin{proof}
    Assume that $H$ and $K$ are GWH-equivalent. Corollary~\ref{cor:prod_GWH_two_orthogonal} implies that $H$ and $K$ have the same weight. i.e., $w_1=w_2$. 
    Now, suppose $w=w_1=w_2$. Then,  
    $KK^T = wI_n = HH^T$, so 
    $$K = (HH^T)(K^T)^{-1}= H(H^T(K^T)^{-1}).$$
    According to Theorem~\ref{thm:inverse_of_GWH_matrix},
    $H^T$ and $(K^T)^{-1}$ are GWH matrices of weight $w$ and $w^{-1}$, respectively.
    Now,  Theorem~\ref{thm:prod_of_GWH_matrices2} implies that $H^T(K^T)^{-1}$ is a GWH matrix of weight $ww^{-1}=1$, i.e., an orthogonal matrix. Hence, taking $Q = H^T(K^T)^{-1}$, yields the result.
    \end{proof}

The situation is very different if we consider GWH matrices of weight $w=0$ since these matrices are not necessarily GWH-equivalent. 
\begin{example}
Consider the matrices 
$$
A=
\begin{bmatrix}
    1 & 1\\
    0 & 0
\end{bmatrix}
\text{ and }
B=
\begin{bmatrix}
    0 & 0\\
    0 & 0
\end{bmatrix}.
$$
Both are GWH matrices of weight $0$ over $\Fset_2$. However, there are no orthogonal matrices $P,Q \in \mathsf{Mat}_{2}(\Fset_2)$ such that $B=QAP$, since they have different ranks. 
It follows immediately that, since the matrices are not GWH-equivalent, they cannot be strongly GWH-equivalent.
\end{example}

In the rest of the section we focus on non-zero weight GWH matrices. Let us write $\mathcal{GWH}_n(\Fset_q)^*$ to denote the set of all GWH matrices of order $n$ over $\Fset_q$ with weight $w \neq 0$. 
The next theorem exhibits the algebraic structure of this set.

\begin{theorem}\ 
\begin{enumerate}
    \item[a)] $\mathcal{GWH}_n(\Fset_q)^*$ forms a group with the matrix multiplication.

    \item[b)] The set of orthogonal matrices $\mathcal{O}_n(\Fset_q)$ is a normal subgroup of $\mathcal{GWH}_n(\Fset_q)^*$.

    \item[c)] The quotient group is given by
    $$\mathcal{GWH}_n(\Fset_q)^*/\mathcal{O}_n(\Fset_q) \cong \begin{cases}
        \Fset_q^* & \text{ if $n$ is even, } \\
        \QR(q) & \text{ if $n$ is odd. }
    \end{cases}$$
\end{enumerate}
\end{theorem}
\begin{proof}\ 
  \begin{enumerate}
      \item[a)] It follows directly from Theorems~\ref{thm:prod_of_GWH_matrices2} and \ref{thm:inverse_of_GWH_matrix}.
      
      \item[b)] Consider $H \in \mathcal{GWH}_n(\Fset_q)^*$ and $Q\in \mathcal{O}_n(\Fset_q)$. Then, Theorems~\ref{thm:prod_of_GWH_matrices2} and \ref{thm:inverse_of_GWH_matrix} imply that
      \[
      (HQH^{-1})(HQH^{-1})^T=I_n 
      \]
       Therefore, $H \mathcal{O}_n(\Fset_q)H^{-1} \subseteq  \mathcal{O}_n(\Fset_q)$ for any  $H \in \mathcal{GWH}_n(\Fset_q)^*$. 
       
      \item[c)] 
      Consider $f:\mathcal{GWH}_n(\Fset_q)^* \longrightarrow \Fset_q^*$ such that for any GWH matrix $H$ of weight $w \neq 0$, $f(H)=w$. As a consequence of Theorem~\ref{thm:prod_of_GWH_matrices2}, $f$ is a group homomorphism. Moreover, $\ker{(f)}= \mathcal{O}_n(\Fset_q)$, so the First Isomorphism Theorem implies that
      $$ 
      \mathcal{GWH}_n(\Fset_q)^*/ \mathcal{O}_n(\Fset_q) \cong \im{(f)}.
      $$
      From Theorem~\ref{thm:existance}, we have $\im{(f)}=\Fset_q^*$ if $n$ is even, or $\im{(f)}=\QR(q)$, otherwise.
  \end{enumerate}
\end{proof}

Observe that, while the group $\mathcal{GWH}_n(\Fset_q)^*$ is not abelian, its quotient $\mathcal{GWH}_n(\Fset_q)^*/ \mathcal{O}_n(\Fset_q)$ is. This behavior is quite uncommon in matrix group theory. Notice also that, from Definition~\ref{def:eq}, we have
\[
\mathcal{GWH}_n(\Fset_q)^* / \mathcal{O}_n(\Fset_q) = \mathcal{GWH}_n(\Fset_q)^* / \sim_{s}.
\]
The next theorem we provide canonical representatives of any equivalence class of the quotient group.

\begin{theorem}\label{thm:canonical_representatives}\
Let $H \in \mathcal{GWH}_n(\Fset_q)^*$ of weight $w \neq 0$. 
\begin{enumerate}
    \item[a)] If $n$ be an even number, then take $a,b \in \Fset_q$ such that $w=a^2+b^2$ and consider \newline
    \begin{equation*} \label{mat:R}
    R:=
    \left[ 
    \begin{array}{rr}
        a &  -b \\
        b & a
    \end{array}
    \right] \otimes I_{\frac{n}{2}}.
   \end{equation*}
    Therefore, $H \sim_s R$.
    \item[b)] If $n$ be an odd number and $w \in \QR(q)$, then consider
     \begin{equation*} \label{mat:D}
    D:= \sqrt{w}I_{n}.
 \end{equation*}
 It holds that $H \sim_s D$.
\end{enumerate}
    
\end{theorem}
\begin{proof}
As it was established in the proof of Theorem~\ref{thm:existance}, the matrices $R$ and $D$ belong to $\mathcal{GWH}_n(\Fset_q)^*$.
Both results are an immediate consequence of Theorem~\ref{th:equiv} and Lemma~\ref{lem:Xaro}.    
\end{proof}

Notice that the canonical representatives exhibited in Theorem~\ref{thm:canonical_representatives} are not unique in general, since equations $x^2 + y^2 = w$ and $x^2 = w$ can have multiple solutions. 

\color{black}


\section{Applications}
\label{sec:ap}

Generalized Weighing-Hadamard matrices over finite fields are important algebraic structures with key applications in combinatorics and information theory \cite{Colbourn2006bk}. Defined by the condition $HH^T = wI_n$ over $\mathbb{F}_q$, these matrices are closely related to coding theory and duality \cite{Mokhtari2026}. Specifically, the row-orthogonality of GWH matrices provides a straightforward framework for constructing linear codes, such as self-dual and complementary dual codes, which optimize error-correction efficiency \cite{Macwilliams1977}. 
Additionally, their structural uniformity makes them useful in digital communications for designing orthogonal spreading sequences \cite{Seberry2005} and in the design of cryptographic primitives  
\cite{Pehlivanoglu2018}.
In this section, we show briefly some possible constructions of linear codes employing GWH matrices.

Recall that an $[n,k]-$\emph{linear code} $\mathcal{C}$ is a $k$-dimensional subspace of $\F_q^n$. The \emph{dual code} of $\mathcal{C}$ is defined by $\mathcal{C}^{\perp}=\{\vec{x} \in \mathcal{C}:  \langle\vec{x},\vec{y}\rangle = 0 \mbox{ for all }  \vec{y} \in \mathcal{C}\}$,  where $\langle\vec{x},\vec{y}\rangle=\sum_{i=1}^nx_iy_i$ is the Euclidean inner product. Moreover, if $\mathcal{C}\subset \mathcal{C}^{\perp}$, then $\mathcal{C}$ is \emph{self-orthogonal}; and if $\mathcal{C}=\mathcal{C}^{\perp}$, then $\mathcal{C}$ is \emph{self-dual}. When a code $\mathcal{C}$ satisfies $\mathcal{C}\cap \mathcal{C}^{\perp} =\{\textbf{0}\}$, it is called \emph{Linear Complementary Dual Codes} (LCD). 
Self orthogonal codes and LCD codes have many applications in cryptography,
communication systems, data storage, and quantum coding theory \cite{Carlet2015,Rawat2014,Guenda2018, Grassl2023}.


\begin{theorem}\cite{Massey1992}
A linear code $\mathcal{C}$ with generator matrix $G$ is LCD if and only if $GG^T$ is invertible.
\end{theorem}

The next results depict methods to construct self-orthogonal (self-dual) or LCD starting from a GWH matrix.

\begin{theorem}
    Consider $G \in \mathsf{Mat}_{n}(\Fset_q)$ a GWH matrix of weight $w$ and for $k<n$, the submatrix $G_k \in \mathsf{Mat}_{k \times n}(\Fset_q)$ given by 
    $k$ arbitrary rows of $G$.
    \begin{enumerate}
        \item [a)] If $w=0$ and $\rank(G_k)=k$, then $G_k$ is a generator matrix of a self-orthogonal code. In particular, if $n$ is even and $k=\frac{n}{2}$, then $G_k$ is a generator matrix of a self-dual code.
        \item [b)] If  $w \neq 0$, then $G_k$ is a generator matrix of a LCD code.
        \end{enumerate}
\end{theorem}
\begin{proof}
    The proof is an immediate consequence of the fact that $
     G_kG_k^T=wI_k.$  
\end{proof}

\begin{theorem}
Let $A \in \mathsf{Mat}_{n-k}(\mathbb{F}_q)$ be a GWH matrix of weight $w$ and, for $k\leq n-k$, the submatrix $A_k \in \mathsf{Mat}_{k \times (n-k)}(\Fset_q)$ given by the first $k$ rows of $A$. Consider $\mathcal{C}$ an $[n,k]$-linear code with generator matrix $G = [I_k \ | \ A_{k}]$. 
\begin{enumerate}
       \item [a)] If $w = -1$, then $\mathcal{C}$ is a self-orthogonal code. In particular, if $k=\frac{n}{2}$, then $\mathcal{C}$ is self-dual.
       
        \item [b)] If $w \neq -1$, then $\mathcal{C}$ is an LCD code.
\end{enumerate}
\end{theorem}
\begin{proof}
    It is an immediate consequence of the fact that 
    $
     GG^T= I+A_kA_k^T= (1+w)I_k.   
    $
\end{proof}


Theorem~\ref{thm:existance} provides gives a proof of next theorem, which is Theorem~4.1.12 of \cite{Huffman2021}.

\begin{theorem} Let $\mathcal{C}$ be $[n,k]$ self-dual code over $\mathbb{F}_q$ that admits a generator matrix in systematic form. Then either $-1$ is a quadratic residue over $\mathbb{F}_q$ or $k$ is even.    
\end{theorem}

\begin{proof} If $\mathcal{C}$ admits a generator matrix in systematic form, i.e., $G = [I_k \mid A]$, then $AA^T = -I_k$, so $A \in \mathsf{Mat}_k(\mathbb{F}_q)$ is a GWH matrix of weight $w = -1$. Theorem~\ref{thm:existance} implies that $-1$ is a quadratic residue or $k$ is even.   
\end{proof}



\section{Conclusions}
In this work, we study the family of GWH matrices.
These matrices generalize both classical Hadamard matrices and the $H$-type matrices introduced by Kojima.

We establish an existence theorem for GWH matrices and show that the invertible GWH matrices form a group under matrix multiplication. Moreover, we prove that the quotient of this group by the subgroup of orthogonal matrices is abelian. We also introduce a notion of equivalence for GWH matrices and provide several constructions. In particular, we present a construction of circulant GWH matrices. We also show how they can be used to construct self-orthogonal, self-dual and LCD codes.

As future work, we intend to investigate non-square GWH matrices, their applications to coding theory and cryptography, and possible extensions of these constructions over rings.
In particular, we plan to study further connections between LCD codes and GWH matrices in the construction of quantum error-correcting codes.

\section{Acknowledgments}

The first author was supported by project FAPEMIG RED-00133-21. The work of the third author was    financed in part by the Conselho Nacional de Desenvolvimento Científico e Tecnológico (CNPq) - Process Number 405842/2023-6 and
by the São Paulo Research Foundation (FAPESP), Brazil, Processes 2024/05051-7 and 2024/00923-6. 
The second and fourth author were partially supported by the Spanish I+D+i project PID2022-142159OB-I00 of the Ministerio de Ciencia e Innovaci\'{o}n, I+D+i project CIAICO/2022/167 of the Generalitat Valenciana, and the I+D+i project VIGROB23-287 and UADIF23-132 of the University of Alicante.











\bibliography{sn-bibliography}

\end{document}